%% file: QACRP_20130519_IsseiSato.tex
\documentclass[authoryear,preprint,3p,twocolumn]{elsarticle}
\usepackage{url}
\usepackage{graphicx}
\usepackage[round]{natbib}
\usepackage{amssymb}
\usepackage{amsmath}
\usepackage{algorithm}
\usepackage{algorithmic}
\usepackage{mathrsfs}

\include{macros}
\include{macros_theorems}

\usepackage{amssymb}

\journal{Neurocomputing}

\begin{document}

\begin{frontmatter}



\title{Quantum Annealing for Dirichlet Process Mixture Models with Applications to Network Clustering}

 \author[label1]{Issei Sato}
 \author[label2]{Shu Tanaka}
 \author[label3]{Kenichi Kurihara}
 \author[label4]{Seiji Miyashita}
 \author[label1]{Hiroshi Nakagawa}
 \address[label1]{Information Technology Center, University of Tokyo, 7-3-1, Hongo, Bunkyo-ku, Tokyo, 113-0033, Japan}
 \address[label2]{Department of Chemistry, University of Tokyo, 7-3-1, Hongo, Bunkyo-ku, Tokyo, 113-0033, Japan}
 \address[label3]{Google, 6-10-1, Roppngi, Minato-ku, Tokyo, 106-6126, Japan}
 \address[label4]{Department of Physics, University of Tokyo, 7-3-1, Hongo, Bunkyo-ku, Tokyo, 113-0033, Japan
}

\begin{abstract}
We developed a new quantum annealing (QA) algorithm for Dirichlet process mixture (DPM) models based on the Chinese restaurant process (CRP).
QA is a parallelized extension of simulated annealing (SA), i.e., it is a parallel stochastic optimization technique.
Existing approaches \citep{Kurihara:UAI2009,Sato:UAI2009} cannot be applied to the CRP because their QA framework is formulated using a fixed number of mixture components.
The proposed QA algorithm can handle an unfixed number of classes in mixture models.
We applied QA to a DPM model for clustering vertices in a network where a CRP seating arrangement indicates a network partition.
A multi core processer was used for running QA in experiments, the results of which show that QA is better than SA, Markov chain Monte Carlo inference, and beam search at finding a maximum a posteriori estimation of a seating arrangement in the CRP.
Since our QA algorithm is as easy as to implement the SA algorithm, it is suitable for a wide range of applications.
\end{abstract}


\end{frontmatter}



\section{Introduction}\label{sec:intro}
Clustering is one of the most important topics	in machine learning because it is a fundamental approach to analyze differences and similarities of data.
In statistical machine learning, a probabilistic latent variable model is used for clustering.
The Dirichlet process mixture (DPM) models \citep{Antoniak74} are well studied and
they enable us to handle an unfixed number of classes, which means that we do not have to decide the number of classes in advance.
In other words, they can estimate the number of classes according to data.
A DPM model is often represented by the Chinese restaurant process (CRP) \citep{Aldous1985}, in which clustering is represented as a seating arrangement of customers in a restaurant.
This representation is a useful one helping us understand the clustering process in DPM models.

A clustering problem using a probabilistic model is generally formulated as a maximum a posteriori (MAP) estimation in statistical machine learning.
Since finding the exact MAP solution will be difficult in many cases, we have to search for an approximate one.
Markov chain Monte Carlo (MCMC) inference is widely used for the CRP \citep{Neal00}
but the MCMC is not necessarily appropriate for the MAP estimation.
When we use MCMC for the MAP estimation, we extract a single class assignment with the highest probability in the class assignments sampled from the posterior distribution.
The problem is that this sampling distribution (i.e., the posterior distribution) has to be stationary, and much iteration is needed before it converges.

\cite{Hal:AISTATS2007} showed that a beam search provides an attractive alternative to the MCMC in the CRP and another approach for the MAP estimation is a stochastic search.
One of the most well-known stochastic search algorithms is simulated annealing (SA) \citep{Kirkpatrick83}, which is similar to the MCMC but has an additional parameter,	called a temperature,  controlling the uncertainty of the search space.
SA is known to find the global optimum when the cooling temperature reduction schedule is slow enough \citep{Geman84} but such a schedule is too slow for practical use.
SA with a practical cooling schedule is therefore also affected by a local optimization problem.

In this work, we focus on a novel stochastic search algorithm, quantum annealing (QA), which  has attracted attention as an alternative annealing method for optimization problems \citep{Kadowaki98,Farhi2001,Santoro02} in quantum information science \citep{Lloyd1996,NielsenChuang2000}.
QA has been shown experimentally converge faster than to find better local optimums for Ising spin models.
It has a parameter inducing quantum fluctuation, so the search space is controlled in a way different from that in SA.
The details are explained below.

\begin{figure}[t!]
\begin{center}
\includegraphics[width=8cm]{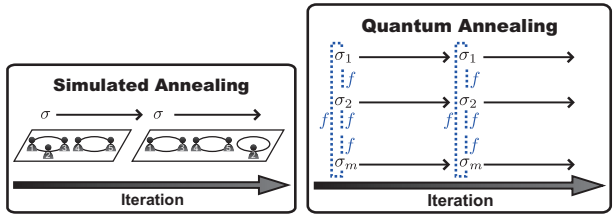}
\end{center}
\caption{The left-hand panel shows the running of SA, in which $\sigma$ indicates a seating arrangement of $N$ customers in the CRP (i.e., a class assignment of $N$ data points).
The right-hand panel shows QA, in which multiple SAs interact through $f$. $\sigma_j$ indicates a seating arrangement of the CRP running in the $j$-th process.
Note in QA that $\sigma_m$ is interacted with $\sigma_{m-1}$ and $\sigma_1$ (i.e., $\sigma_{m+1}=\sigma_1$), which is mathematically derived from the QA framework (Theorem \ref{theorem:qast}). During iterations, we control the hyper-parameters.}
\label{qa_exp}
\end{figure}

QA is a parallelized extension of SA in which quantum fluctuation is induced by running multiple SAs with interactions.
Let us consider running $m$ SAs, and let $\sigma_j~(j=1,\cdots,m)$ indicate a state (e.g., a class assignment) of $N$ data points in the $j$-th simulation.
In the CRP, we formulate $\sigma_j$ as a table-seating arrangement of customers in the $j$-th CRP (see	Fig. \ref{qa_exp}).
We denote $N$ data points as $\bx=(x_1,x_2,\cdots,x_N)$.
The log-likelihood model given $\sigma_j$ (i.e., $\log p(\bx,\sigma_j)$) is based on the way the data is modeled.
For simplicity, we denote the log-likelihood by $\log p(\sigma_j)$.
In this work, we use the Newman model \citep{newman:2007} for clustering network data (explained in Sec.\ref{experiments}).
QA runs multiple {\it dependent} SAs with {\it dependent} here meaning that there is interaction $f$ among neighboring SAs (see right-hand panel in Fig. \ref{qa_exp}).

We describe QA in terms of an optimization problem.
When we run $m$ SAs with different random initializations independently, we optimize $\log p(\sigma_j)$ individually.
That is, we find  $\sigma_j^* = \argmax_{\sigma_j} \log p(\sigma_j)$ for each $j$ and we choose $\sigma$ that has the highest $\log p(\sigma_j^*)$ of all $j$.
In QA, we optimize the joint probability of $m$ CRPs' states $\{\sigma_j\}_{j=1}^m$:
\begin{align}
\max_{(\sigma_1,\sigma_2,\cdots,\sigma_m)} \log p_{\rm QA}(\{\sigma_1,\sigma_2,\cdots,\sigma_m\}), \label{qa_opt}
\end{align}
where $p_{\text{QA}}(\cdot)$ is a probability measure over a set of states, which means that each state $\sigma_j~(j=1,\cdots,m)$ can take an independent state and QA gives the probability for these states.
A set of states $(\sigma_1,\sigma_2,\cdots,\sigma_m)$ represents (quantum) superposition of different states.
That is,  $p_{\text{QA}}(\cdot)$ is a probability measure over superposition of different states in the limit of $m \rightarrow \infty$ in quantum physics.
In the CRP, $(\sigma_1,\sigma_2,\cdots,\sigma_m)$ represents a superposition of $m$ seating arrangements.
The optimization problem \eqref{qa_opt} is actually formulated as
\begin{align}
\max_{(\sigma_1,\sigma_2,\cdots,\sigma_m)} \sum_{j=1}^m \log p_{\rm SA}(\sigma_j) + f \cdot R(\sigma_1,\sigma_2,\cdots,\sigma_m), \label{qa_opt2}
\end{align}
where the first term corresponds to the summation over $m$ SA objectives, and $R(\cdot)$ is regarded as a regularizer among $m$ states, which are described in Sec.\ref{seq:approximation:qacrp}.
This optimization is derived from the QA framework explained in Sec. \ref{sec:qacrp}.
QA was recently used for solving practical optimization problems, such as clustering \citep{Kurihara:UAI2009} and variational Bayes inference \citep{Sato:UAI2009}, and it outperformed SA.
Figure \ref{target} summarizes QA and related work.

\begin{figure}[t!]
\begin{center}
\includegraphics[width=8cm]{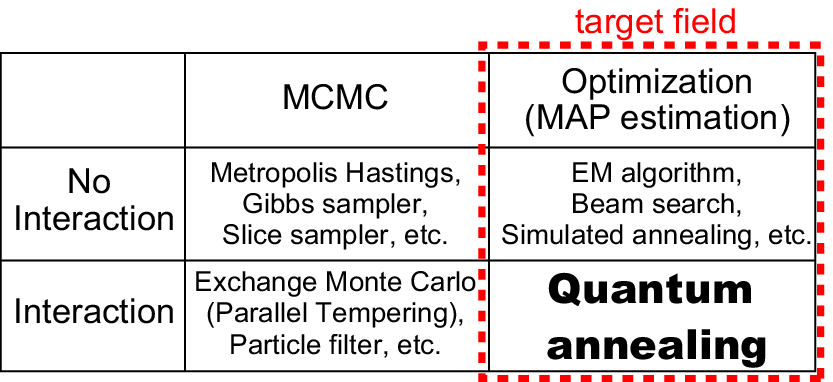}
\end{center}
\caption{QA and related algorithms. We categorize the algorithms into two groups according to their purposes: MCMC and optimization.
MCMC aims at approximating the expectation by a finite sum of samples drawn from the posterior distribution and therefore needs a large number of iterations.
Optimization algorithms search for optimal solutions in a small number of iterations.
The algorithms are also classified according to the presence of interactions among multiple processes. Note that expectation maximization (EM) algorithms \citep{DempsterLairdRubin77} and variational inferences \citep{Blei05,Kurihara07IJCAI} cannot be applied to the CRP.}
\label{target}
\end{figure}

{\bf Problems:}
Existing approaches \citep{Kurihara:UAI2009,Sato:UAI2009} cannot be applied to the CRP because they need a fixed number of mixture components.
Moreover, these approaches have to use a heuristic such as purity to
apply their QA algorithms to the clustering problem.
Therefore, a different formulation is needed.

{\bf Contributions:}
The purpose of this study is to propose a QA algorithm for the CRP.
The key point is how to represent the states of data in the CRP.
The existing work represents the data states  as ``which class a data point is assigned to.''
That is, they require $K$-dimensional indicator vectors to represent the data states, and $K$ (the number of classes) is given and fixed.
We instead represent the states of data as ``which data points a data point shares the table with'' in the CRP.
That is, we use an $N$-by-$N$ bit matrix to represent the data states, and this matrix indicates a seating arrangement in the CRP and does not depend on $K$.
This bit-matrix representation of the CRP is a novel idea and a key point in applying QA to the CRP.
Note that the bit matrix is only used for mathematically deriving QA for the CRP and is not used in the actual algorithm.
Mathematically, this novelty appears in interaction function $f$ in Eq.\eqref{eq:f}, where our derived $f$ does not include the number of classes, $K$, whereas $f$ in existing work is formulated by using (fixed) $K$.
Moreover, our algorithm does not require  heuristics such as purity.
We also use parallel processing in QA, whereas \cite{Kurihara:UAI2009} and \cite{Sato:UAI2009} used a single processing.

The idea of using a matrix that represents the relationship among data points has been used in several studies \citep{KrzyzakOja:NN1993,FreyDueck:Science2007,WangLai:NC2011}.
These studies used a similarity matrix based on the feature vectors among data points.
For example, $x_i$ and $x_j$ denote the feature vectors of the i-th and j-th data points and the similarity was calculated by using the Gaussian kernel between $x_i$ and $x_j$. The formulation in this paper is different from these existing approaches because we do not use the similarity matrix based on the feature vectors. We used the matrix of data points to formulate the clustering state of data in a learning process.
For example, a bit-matrix in this study changes in a learning process, while the similarity (kernel) matrix in the existing works does not change but is fixed. Moreover, we use a bit-matrix only for mathematically deriving QA for the CRP.
We do not directly use the matrix in an actual algorithm, whereas the existing approaches directly use the similarity matrix for clustering data.


%
\begin{figure*}[t!]
\begin{center}
\includegraphics[width=14cm]{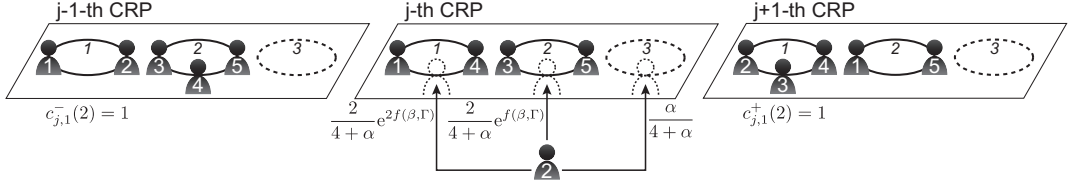}
\end{center}
\caption{
Example of approximation inference in QACRP where $\beta=m$.
Let us consider adding customer $2$ to restaurant $j$ ($j$-th CRP).
The classical CRP seats customers at existing tables in proportion to the number of customers already seated (see Eq.\eqref{eq:py:crp}). The QACRP sampler derived in Eq.\eqref{eq:py:qacrp} introduces the effect of customers who share tables with customer $2$ in the ($j-1$)-th and ($j+1$)-th CRPs.
In this case, since customers $1$, $3$, and $4$ are customers who share tables with customer $2$ in the ($j-1$)-th and ($j+1$)-th CRPs, the $1$st table  in the $j$-th CRP has these ``two'' customers and so takes the effect $e^{\text{``}2\text{''} f(\beta,\Gamma)}$ into account ($c_{j,1}^{-}(2)+c_{j,1}^{+}(2)=2$ in Eq. \eqref{eq:py:qacrp}).
That is, in the QACRP sampler, when interaction $f(\beta,\Gamma)$ is large, a customer tends to sit with customers sharing the table in other CRPs.
}
\label{fig:qacrp}
\end{figure*}

\section{Chinese Restaurant Process (CRP)}
The CRP is a distribution over partitions such as clustering and is composed of three elements: a customer, table, and restaurant.
In a clustering problem, the customer denotes a data point and the table denotes a data class.
A seating arrangement of customers in a restaurant indicates a class assignment of data.
In QA, we run multiple CRPs, i.e., we consider the seating arrangements in multiple restaurants.

The CRP assigns a probability for the seating arrangement of the customers in which
$\bZ=\{z_{i}\}_{i=1}^{N}$ denotes the seating arrangement of the customers and
$z_{i}=k$ indicates that customer $i$ sits at the $k$-th table.
$N$ indicates the number of customers.
When customer $i$ enters a restaurant with $K$ occupied tables at which other customers are already seated, customer $i$ sits at a table with the following probability:
{\small
\begin{align}
&p(z_{i}|\bZ  \backslash {z_i};\alpha) \propto
\left \{
\begin{array}{ll}
 \displaystyle \frac{N_{k}}{\alpha+N-1}&(\text{$k$-th occupied table}),\\
 \displaystyle \frac{\alpha}{\alpha+N-1}&( \text{new unoccupied table}),
\end{array}
\right . \label{eq:py:crp}
\end{align}
}
where $N_{k}$ denotes the number of customers sitting at the $k$-th table, and $\alpha$ is the hyper parameter of the CRP.
A customer tends to select a new table when $\alpha$ takes large value.

The log-likelihood of $\bZ$ is given by
%
$
p(\bZ)=\frac{\alpha^{K(\bZ)}}{\prod_{l=1}^N (N-l+\alpha)}\prod_{k=1}^{K(\bZ)}(N_k-1)!,
$
where $K(\bZ)$ is the number of occupied tables in $\bZ$.


\section{Quantum Annealing for CRP}\label{sec:qacrp}
This section explains how we derive QA for the CRP (QACRP).
First, we introduce some notations and explain QACRP intuitively.
Second, we introduce a bit matrix to reformulate the CRP for using QA independent of the number of classes, which is a key idea in this work.
Third, we formulate the CRP by using a ``density matrix'' that is a basic formulation in quantum mechanics.
Finally, we apply the Suzuki-Trotter expansion \citep{Trotter59,Suzuki76} to approximate QACRP	because the first-derived QACRP is intractable because of the computational cost of a matrix exponential.

\subsection{Main result (See Fig. \ref{fig:qacrp} for an intuitive image)}
QACRP uses multiple restaurants.
$z_{j,i}=k$ indicates that customer $i$ sits at the $k$-th table in the $j$-th restaurant.
$\bZ_j=\{z_{j,i}\}$ denotes the seating arrangements of customers in the $j$-th restaurant.
$c_{j,k}^{+}(i)$ denotes the number of customers who sit at the $k$-th table in the $j$-th restaurant and share tables with  customer $i$ in the ($j+1$)-th restaurant.
$c_{j,k}^{-}(i)$ denotes the number of customers who sit at the $k$-th table in the $j$-th restaurant and share tables with  customer $i$ in the ($j-1$)-th restaurant.
Customer $i$ sits at a table in the $j$-th restaurant with the following probability:
\begin{align}
&p_\text{QA}(z_{j,i}|\{\bZ_d\}_{d=1}^{m}  \backslash \{z_{j,i}\};\beta, \Gamma) \propto \nn\\
&\begin{cases}
 \displaystyle \left ( \frac{N_{j,k}}{\alpha+N-1}\right )^{\frac{\beta}{m}} e^{(c_{j,k}^{-}(i)+c_{j,k}^{+}(i))f(\beta,\Gamma)}\\
 ~~~~~~~~~~~~~~~~~~~~~~~~~~~~~~~~~~~~(\text{$k$-th occupied table}),\\
 \displaystyle \left ( \frac{\alpha}{\alpha+N-1} \right )^{\frac{\beta}{m}}
 ~~~~~~~~~~~~~~( \text{new unoccupied table}),
\end{cases}
\label{eq:py:qacrp}
\end{align}
where $N_{j,k}$ denotes the number of customers sitting at the $k$-th table in the $j$-th restaurant.
$f(\beta,\Gamma)$ is derived in Sec.\ref{seq:approximation:qacrp} Eq.\eqref{eq:f} where $\beta$ and $\Gamma$ are hyper parameters that are called inverse temperature and quantum effect, respectively.
The inverse temperature is also the hyper parameter of SA.
When you change the CRP in Eq.(\ref{eq:py:crp}) into QACRP in Eq.(\ref{eq:py:qacrp}), all you have to do is to count the customers sharing tables in neighboring CRPs and introduce $f(\beta,\Gamma)$.
Figure \ref{fig:qacrp} shows an example of QACRP and provides an intuitive explanation.
When $f(\beta,\Gamma)=0$, QACRP is equivalent to $m$ independent CRPs with inverse temperature $\beta/m$, which we call SACRPs.
We provide the details of the derivation in the next sections.


\begin{figure*}[t!]
\begin{center}
$
\includegraphics[width=14cm]{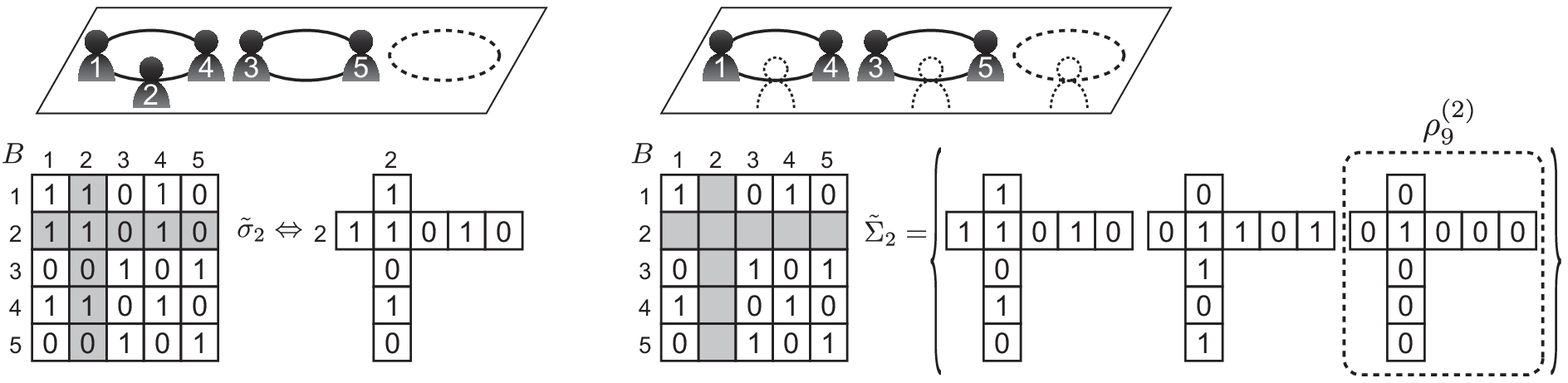}
$
\end{center}
\caption{Example of bit matrix representation.
A seating arrangement $\bZ$ is represented as a bit matrix $\bB$, which enables us to formulate the CRP without fixing the number of tables.
For example, $\tilde \sigma_{2}$ represents customers who share a table with customer $2$.
$\tilde \Sigma_{2}$ represents a set of the states that customer $2$ can take under the seating conditions, and
$\rho_{9}^{(2)}$ indicates that customer $2$ sits alone at a table.
} \label{fig:bm_rep}
\end{figure*}

\subsection{Bit matrix representation for CRP}

%

We represent seating arrangement $\bZ$ by using a bit matrix $\bB$ in order to reformulate the CRP without fixing the number of tables (see Contributions in Sec. \ref{sec:intro}).
Although this bit matrix representation seems to have high computational complexity,  in an actual algorithm of QA, we do not need the direct calculation to the bit matrix.

A bit matrix $\bB$ looks like an adjacency matrix of customers (see Fig. \ref{fig:bm_rep}) and
$\bB$  denotes an $N$-by-$N$ bit matrix where $B_{i}$ is the $i$-th row vector, i.e., $B_{i}=(B_{i,1},B_{i,2},\cdots,B_{i,N})$,
and $B_{i,n}$ is the $i$-th row and the $n$-th column element of $\bB$ or the $n$-th element of $B_i$.
$\tilde \sigma_{i,n}~(i,n=1,\cdots,N)$ is a two-dimensional indicator vector, i.e., it takes $(1,0)^\top$ or $(0,1)^\top$
We correspond $B_{i,n}=1$ to $\tilde \sigma_{i,n}=(1,0)^{\top}$ and $B_{i,n}=0$ to $\tilde \sigma_{i,n}=(0,1)^{\top}$, which means we can represent $\bB$ by using the $2^{N^2}$ dimensional indicator vector, $\sigma$, as follows:
\begin{align}
\bB \Leftrightarrow \sigma=\bigotimes_{i=1}^{N}\bigotimes_{n=1}^{N} \tilde \sigma_{i,n}.\label{def:B}
\end{align}
$\bigotimes$ is the Kronecker product, which is a special case of a tensor product.
If  $A$ is a $k$-by-$l$ matrix and $B$ is an $m$-by-$n$ matrix,
then  the Kronecker product $A \bigotimes B$ is the following $km$-by-$ln$ block matrix:
$A \bigotimes B =\left ( \begin{matrix} a_{11} B & \cdots & a_{1l}B \\ \vdots & \ddots & \vdots \\ a_{k1} B & \cdots & a_{kl} B \end{matrix} \right )$.
For example, $(1,0)^\top \bigotimes  (0,1)^\top =(0,1,0,0)^\top$.
$\Sigma$ denotes a set of $\sigma$, i.e., $|\Sigma|=2^{N^2}$.

The bit matrix $\bB$ is regarded as a seating arrangement as follows.
If $B_{i,n}=1$, the $i$-th and the $n$-th customers share a table.
Note that we need the following conditionals to represent seating arrangements with the bit matrix
\begin{enumerate}
\item $B_{i,n}=B_{n,i}$ (symmetric matrix)
\item $B_{i,i}=1(i=1,2,\cdots,N)$, i.e., $\Tr{B}=N$
\item $\forall$ $i$ and $l$,  $\frac{B_{i}}{|B_{i}|}\cdot \frac{B_{l}}{|B_{l}|}= 1$ or $0$, where $\cdot$ is the inner product.
\end{enumerate}
$\Tr$ $X$ is the trace of $X$.
$\tilde \Sigma(\subset \Sigma)$ denotes a set of $\sigma$ corresponding to $\bB$ satisfying the above conditions.
We call these conditions ``seating conditions.''

Here, $\tilde \sigma_{i}$ indicates the state of the $i$-th customer, i.e., with whom the $i$-th customer shares a table (see the left-hand side of Fig. \ref{fig:bm_rep}) and
$\tilde \sigma_{i}$ is a ($2N-1$ )-dimensional indicator vector given by
\begin{align}
\tilde \sigma_{i}=\bigotimes_{n=1}^{i-1}\tilde \sigma_{i,n} \otimes \tilde \sigma_{i,i} \otimes \bigotimes_{n=i+1}^{N}\tilde \sigma_{i,n} \bigotimes_{n=1,n\not =i}^{N} \tilde \sigma_{n,i}.
\end{align}

Let $\tilde \Sigma_{i}$ be a set of the states that $\tilde \sigma_{i}$ can take under the seating conditions (i.e., $\sigma \in \tilde \Sigma$) when the $i$-th row elements and the $i$-th column elements are blank and the others are filled (see the right-hand side of Fig. \ref{fig:bm_rep}).
Since $\tilde \Sigma_{i}$ is a set of table-assignment states of the $i$-th customer, $|\tilde \Sigma_{i}|=K(\bZ \backslash \{z_i\})+1$.
For example, the right-hand side of Fig. \ref{fig:bm_rep} shows table assignments of the $2$nd customer when customers $1$,$3$,$4$, and $5$ have already been seated.

$\rho_{2N-1}^{(i)}$ is defined as a $2N-1$ dimensional indicator vector given by
\begin{align}
&\rho_{2N-1}^{(i)}=\nn\\
&\bigotimes_{n=1}^{i-1}(0,1)^{\top} \otimes (1,0)^{\top} \otimes \bigotimes_{n=i+1}^{N}(0,1)^{\top} \bigotimes_{n=1,n\not =i}^{N} (0,1)^{\top}.
\end{align}
The right-hand side of Fig. \ref{fig:bm_rep} shows an example of $\rho_{2N-1}^{(i)}$.
We use $\rho_{2N-1}^{(i)}$ only in \ref{proof}.

\subsection{Density matrix representation for classical CRP}
We define the energy function $E$ over $\sigma^{(l)}\in \Sigma ~~(l=1,\cdots,2^{N^2})$ by
$E(\sigma^{(l)})= - \log p(\sigma^{(l)})$,
where if $\sigma^{(l)} \in \Sigma \backslash \tilde \Sigma$, then $p(\sigma^{(l)})=0$, i.e., $E(\sigma^{(l)})=+\infty$.

The probability of a state $\sigma (\in \Sigma)$ is given by
\begin{align}
  p(\sigma) = \frac{1}{Z} \sigma^{\top}e^{-\cH_\text{c}}\sigma,
\end{align}
where  $\cH_\text{c}= {\rm diag} \left[ E(\sigma^{(1)}),E(\sigma^{(2)}),\cdots,E(\sigma^{(2^{N^2})}) \right]$, and ${\rm diag}[\cdot]$ denotes a diagonal matrix.
Note that $Z=\sum_{\sigma} \sigma^{\top}e^{-\cH_\text{c}}\sigma = \Tr {e^{-\cH_\text{c}}}$, where $\cH_\text{c}$ is called the classical Hamiltonian.
If $\sigma \in \tilde \Sigma$, then $p(\sigma)$ is equal to $p(\bZ)$, i.e., $p(\sigma)$ is the probability over a seating arrangement.
Since $\cH_\text{c}$ is diagonal, $e^{-\cH_\text{c}}$ is also diagonal with the $k$-th diagonal element  $e^{-E(\sigma^{(k)})}$.
That is, $p(\sigma^{(k)})=\frac{1}{Z} \sigma^{(k)\top}e^{-\cH_\text{c}}\sigma^{(k)}=\frac{1}{Z}e^{-E(\sigma^{(k)})}$.

\subsection{Formulation for quantum CRP}
The basic approach to expanding a classical system to a quantum one is to make the Hamiltonian non-diagonal, i.e., add some off-diagonal elements while keeping hermiticity.
We define a non-diagonal matrix $\cH$ by
\begin{align}
\cH = \cH_\text{c} + \cH_\text{q},
\end{align}
where $\cH_{\text{q}}$ is a non-diagonal matrix (we describe the definition of $\cH_\text{q}$ later).
Intuitively,  diagonal elements are filled with zero, and some off-diagonal elements are filled with $\Gamma$ in $\cH_{\text{q}}$.
That is, $\cH$ is filled with energy $E(\sigma)$ in diagonal elements and quantum effect $\Gamma$ in off-diagonal elements.
The above scheme that adds a non-diagonal matrix ($\cH_{\text{q}}$) to a diagonal matrix ($\cH_{\text{c}}$) is a basic approach in quantum physics and has also worked well in \citep{Kurihara:UAI2009,Sato:UAI2009}.
The meaning of this formulation was described in \citep{Sato:UAI2009} in terms of uncertainty.

The probability of a state $\sigma (\in \Sigma)$ in a quantum system is given by
\begin{align}
p_\text{QA}(\sigma;\beta, \Gamma) = \frac{1}{Z} \sigma^{\top}e^{-\beta(\cH_\text{c}+\cH_\text{q})}\sigma,\label{eq:p_qa}
\end{align}
where $Z=\sum_{\sigma} \sigma^{\top}e^{-\beta(\cH_\text{c}+\cH_\text{q})}\sigma$=$\Tr [ e^{-\beta(\cH_\text{c}+\cH_\text{q})} ]$.

The optimization problem
\begin{align}
\max_{\sigma}~\log p_\text{QA}(\sigma;\beta, \Gamma) \label{qopt}
\end{align}
could be solved by using the eigenvalue decomposition  of  the density matrix $\frac{1}{Z} e^{-\beta(\cH_\text{c}+\cH_\text{q})}$, but, this approach is intractable because of its large computational cost.

One approximation approach for solving the optimization problem \eqref{qopt} is a stochastic search by drawing a state of the $i$-th customer, $\tilde \sigma_i$, from
\begin{align}
&p(\tilde \sigma_i|\sigma \backslash \tilde \sigma_i) = \frac{\sigma^{\top}e^{-\cH_\text{c}}\sigma}{\sum_{\tilde \sigma_i} \sigma^{\top}e^{-\cH_\text{c}}\sigma},\label{eq:gibbs:dm-crp-sa}\\
&p_\text{QA}(\tilde \sigma_i|\sigma \backslash \tilde \sigma_i;\beta,\Gamma) = \frac{\sigma^{\top}e^{-\beta(\cH_\text{c}+\cH_\text{q})}\sigma}{\sum_{\tilde \sigma_i} \sigma^{\top}e^{-\beta(\cH_\text{c}+\cH_\text{q})}\sigma},\label{eq:gibbs:dm-crp-qa}
\end{align}
where $\sigma \backslash \tilde \sigma_i$ indicates  that bits excluding the $i$-th row and the $i$-th column elements are standing.
The summation over $\tilde \sigma_i$ is actually the summation of $\tilde \sigma_i \in \tilde \Sigma_i$; therefore, the classical system $p(\tilde \sigma_i|\sigma \backslash \tilde \sigma_i)$ is tractable when $p(\tilde \sigma_i|\sigma \backslash \tilde \sigma_i)$ in Eq. \eqref{eq:gibbs:dm-crp-sa} is another expression of Eq. \eqref{eq:py:crp}.
Calculation of the probability of the quantum system $p_\text{QA}(\tilde \sigma_i|\sigma \backslash \tilde \sigma_i;\beta,\Gamma)$, however, is intractable because of the exponential operation of a non-diagonal matrix $\cH = \cH_\text{c} + \cH_\text{q}$.
We therefore need another approach described in Sec.\ref{seq:approximation:qacrp}.

We define $\cH_\text{q}$ as follows.
\newcommand{\unit}{
\left(
  \begin{array}{cc}
	1 & 0\\
	0 & 1
  \end{array}
   \right)
}
\newcommand{\sigmax}{
\left(
  \begin{array}{cc}
	0 & 1\\
	1 & 0
  \end{array}
   \right)
}
{\small
\begin{align}
  \cH_{\text{q}} =& -\Gamma \sum_{i=1}^N \sum_{n=1}^N  \sigma_{i,n}^x,~\bbE = \unit,~\sigma^x = \sigmax,\nn
\end{align}
\begin{align}
  \sigma_{i,n}^x =&
   \left(\bigotimes_{t=1}^{i-1} \bigotimes_{u=1}^{N} \bbE \right)
  \otimes
  \left [
  \left(\bigotimes_{t=1}^{n-1} \bbE \right)
  \otimes
  \sigma^x
  \otimes
  \left(\bigotimes_{t=n+1}^N \bbE \right)
  \right ]
  \nn\\
  &\otimes
   \left(\bigotimes_{t=i+1}^{N} \bigotimes_{u=1}^{N} \bbE \right),
\end{align}
}
where $\Gamma$ is the quantum effect parameter.
This formulation means that diagonal elements are filled with zeros, and some off-diagonal elements are filled with $\Gamma$ in $\cH_{\text{q}}$.
Although other definitions can be considered, we define this formulation so that we can make the derivation of the search algorithm tractable by using an approximation method that is easy to implement.

\subsection{Approximation inference for QACRP}\label{seq:approximation:qacrp}
We cannot calculate $p_\text{QA}(\tilde \sigma_i|\sigma \backslash \tilde \sigma_i;\beta,\Gamma)$ in Eq.\eqref{eq:gibbs:dm-crp-qa} because $\sigma^{\top}e^{-\beta \cH}\sigma$ is intractable because of the non-diagonal matrix $\cH$.
We use the Suzuki-Trotter expansion \citep{Trotter59,Suzuki76} to approximate $p_\text{QA}(\tilde \sigma_i|\sigma \backslash \tilde \sigma_i;\beta,\Gamma)$.

We consider multiple running CRPs in which
$\sigma_j(j=1,\cdots,m)$ indicates the seating arrangement of the $j$-th CRP and represents the $j$-th bit matrix $\bB_j$.
We correspond $B_{j,i,n}=1$ to $\tilde \sigma_{j,i,n}=(1,0)^{\top}$ and $B_{j,i,n}=0$ to $\tilde \sigma_{j,i,n}=(0,1)^{\top}$, which means that we can represent $\bB_j$ as $\sigma_j$ by using Eq.\eqref{def:B}.
We derive the following theorem:

\begin{theorem}\label{theorem:qast}
$p_\text{QA}(\sigma;\beta,\Gamma)$ in Eq.\eqref{eq:p_qa} is approximated by the Suzuki-Trotter expansion as follows:
\begin{align}
&p_\text{QA}(\sigma;\beta,\Gamma)= \frac{1}{Z}\sigma^\top e^{- \beta (\cH_\text{c} + \cH_\text{q})} \sigma\nn\\
&  =
  \sum_{\sigma_j(j\geq 2)}p_\text{QA-ST}(\sigma,\sigma_2,\cdots,\sigma_m;\beta,\Gamma)
  + \mathcal{O} \left( \frac{\beta^2}{m} \right),\label{eq:trot_rep}
\end{align}
where we rewrite $\sigma$ as $\sigma_1$, and
\begin{align}
&p_\text{QA-ST}(\sigma_1,\sigma_2,\cdots,\sigma_m;\beta,\Gamma)\nn\\
&= \prod_{j=1}^m \frac{1}{Z(\beta, \Gamma)}
  e^{- \frac{\beta}{m} E(\sigma_j)}
  e^{f(\beta,\Gamma)s(\sigma_j,\sigma_{j+1}) },\label{eq:p_qast}
\\
&f(\beta,\Gamma)=2 \log \coth \left(\frac{\beta}{m}\Gamma \right),\label{eq:f}\\
&  s(\sigma_j,\sigma_{j+1}) = \sum_{i=1}^{N}\sum_{n=1}^{N}\delta( \tilde{\sigma}_{j,i,n},\tilde{\sigma}_{j+1,i,n}),\label{eq:s}\\
&  Z(\beta, \Gamma)=\left[ \sinh \left(\frac{\beta}{m}\Gamma \right)\right]^{2N}\sum_{\sigma}e^{- \frac{\beta}{m} E(\sigma)}.
\end{align}
\end{theorem}
Note that $\sigma_{m+1}=\sigma_1$.
The proof is given in \ref{proof}.
Note that our derived $f$ in Eq.\eqref{eq:f} does not include the number of classes, $K$, whereas the $f$ in existing work \citep{Kurihara:UAI2009,Sato:UAI2009} is formulated by using a fixed $K$.

Equation \eqref{eq:trot_rep} is interpreted as follows. $p_\text{QA}(\sigma;\beta, \Gamma)$ is approximated by marginalizing out other states $\{\sigma_j\}_{j \geq 2}$ of $p_\text{QA-ST}(\sigma_1,\sigma_2,\cdots,\sigma_m;\beta,\Gamma)$.
As shown in Eq.\eqref{eq:p_qast}, $p_\text{QA-ST}(\sigma_1,\sigma_2,\cdots,\sigma_m;\beta,\Gamma)$ looks like the joint probability of the states of $m$ dependent CRPs.
In Eq.\eqref{eq:p_qast}, $e^{- \frac{\beta}{m} E(\sigma_j)}$  corresponds to the classical CRP with inverse temperature and $e^{f(\beta,\Gamma)s(\sigma_j,\sigma_{j+1})}$ indicates the quantum effect part.
If $f(\beta,\Gamma)=0$, which means CRPs are independent,
$p_\text{QA-ST}(\sigma_1,\sigma_2,\cdots,\sigma_m;\beta,\Gamma)$ is equal to the products of probability of $m$ classical CRPs.
$s(\sigma_j,\sigma_{j+1}) (>0)$ is regarded as a similarity function between the $j$-th and ($j+1$)-th bit matrices. If they are the same matrices, then $s(\sigma_j,\sigma_{j+1})=N^2$.
 In Eq.\eqref{qa_opt2}, $\log p_{\rm SA}(\sigma_j)$ corresponds to $\log e^{- \frac{\beta}{m} E(\sigma_j)}/Z$ and the regularizer term $f \cdot R(\sigma_1,\cdots,\sigma_m)$ is $\log \prod_{j=1}^m e^{f(\beta,\Gamma)s(\sigma_j,\sigma_{j+1})}=f(\beta,\Gamma)\sum_{j=1}^m s(\sigma_j,\sigma_{j+1})$.

Note that we aim at deriving the approximation inference for $p_\text{QA}(\tilde \sigma_i|\sigma \backslash \tilde \sigma_i;\beta,\Gamma)$ in Eq.\eqref{eq:gibbs:dm-crp-qa}.
Using Theorem \ref{theorem:qast}, we can derive Eq.\eqref{eq:py:qacrp} as the approximation inference.
The details of the derivation are provided in \ref{dev:qacrp}.

\section{Experiments}\label{experiments}
We evaluated QA in a real application.
We applied QA to a DPM model for clustering vertices in a network where a seating arrangement of the CRP indicates a network partition.

\subsection{Network Model}
We used the Newman model \citep{newman:2007} for network modeling in this work.
The Newman model is a probabilistic generative network model.
This model is flexible, which enables researchers to analyze observed graph data without specifying the network structure (disassortative or assortative) in advance.

\begin{figure}[t!]
\begin{center}
\includegraphics[width=8cm]{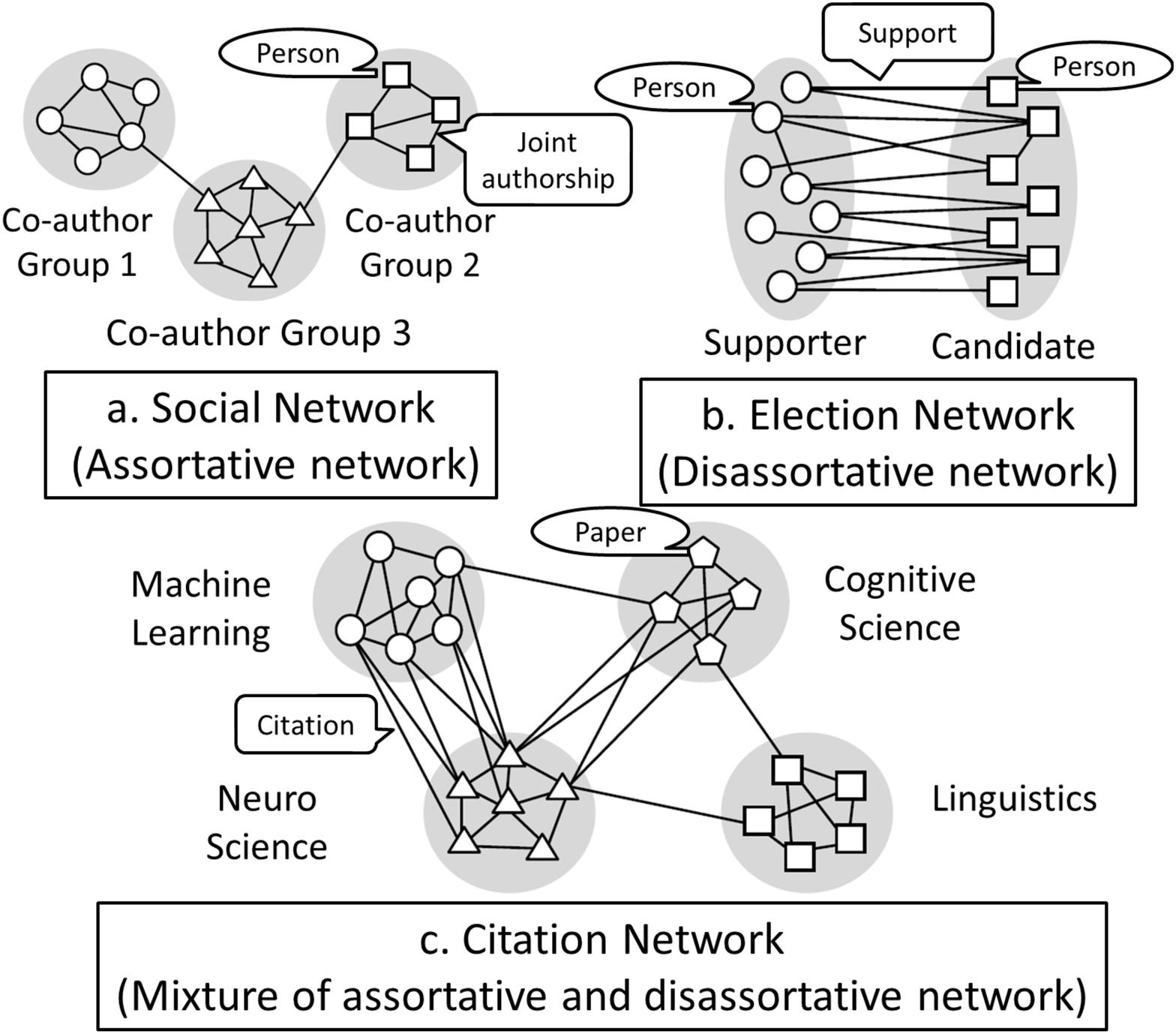}
\end{center}
\caption{Examples of Network structures.}
\label{fig:network}
\end{figure}

In an assortative network, such as a social network, the members (vertices) of each class are mostly connected to the other members of the same class.
The communications between members in three social groups is illustrated in Fig. \ref{fig:network}, where one sees that the members generally communicate more with others in the same group than they do with those outside the group.
In a disassortative network, the members (vertices) have most of their connections outside their class.
An election network of supporters and candidates is illustrated in Fig. \ref{fig:network}-b, where a link indicates support for a candidate.
The Newman model can model not only these two kinds of networks but also a mixture of them, such as a citation network (see Fig.\ref{fig:network}-c),
but, the user must decide in advance the number of classes.
We therefore used the DPM extension of the Newman model as described in \ref{appdx:network_model}.

\subsection{Dataset}
We used three social network datasets, Netscience\footnote{http://www.casos.cs.cmu.edu/computational\_tools\\/datasets/external/netscience/},
Citeseer\footnote{http://www.cs.umd.edu/projects/linqs/projects\\/lbc/index.html}, and Wikivote\footnote{http://snap.stanford.edu/data/wiki-Vote.html}.
Netscience is a coauthorship network of scientists working on a network that has 1,589 scientists (vertices).
Citeseer is a citation network dataset for 2,110 papers (vertices).
Wikivote is a bipartite network constructed  using administrator elections and vote history data in Wikipedia.
Its 7,115 vertices represent Wikipedia users and a directed edge from vertex $i$ to vertex $j$ represents that user $i$ voted for user $j$.
Netscience, Wikivote, and Citeseer respectively correspond to network examples a, b, and c in Fig.\ref{fig:network}.
We used the vertices in these networks to represent customers in the CRP, and
we used a negative log-likelihood as an energy function to find the MAP solution. 

\subsection{Annealing schedule}
We tested several $\beta/m$ schedules using combinations of $\beta_0=0.2m$, $0.4m$, and $0.6m$ and $\beta=\beta_0 \log(1+t)$, $\beta_0 \sqrt{t}$, and $\beta_0 t$, where $t$ denotes the $t$-th iteration.
The results we observed in our experiments showed that $\beta_0=0.4m$ and $\beta_0 \sqrt{t}$ created a better schedule in SA in terms of the MAP estimation.
That is, $\beta$ increases to $\beta_0\sqrt{T}$, where $T$ is the total number of iterations.
In QA, we use the same $\beta/m$ schedule we used in SAs.
Note that since the new table is easy to sample at very small $\beta$ (where the probability distribution becomes flat, see Eq. \eqref{eq:py:qacrp}), the SACRP has many tables at small $\beta$ and converges very slowly.
That is, inverse temperatures that are too low do not work well in the CRP.

Since interaction $f$ is a function of $\Gamma$ and $\beta$, in practice we have to consider the schedule of $f(\beta,\Gamma)$. The interaction $f(\beta,\Gamma)$ increases when $\frac{\beta\Gamma}{m}$ decreases.
QA is known to work well when $f(\beta,\Gamma)$  starts from zero (i.e., ``independent'' multiple SAs) and gradually increases.
This process of $f(\beta,\Gamma)$ is achieved when $\frac{\beta\Gamma}{m}$ is a decreasing function of $t$.
Therefore, we use
\begin{align}
\frac{\beta\Gamma}{m}=\Gamma_0 \frac{T}{t},
\end{align}
where  $\Gamma_0$ is a tuning parameter.

\begin{figure}[t!]
\begin{center}
\includegraphics[width=8cm]{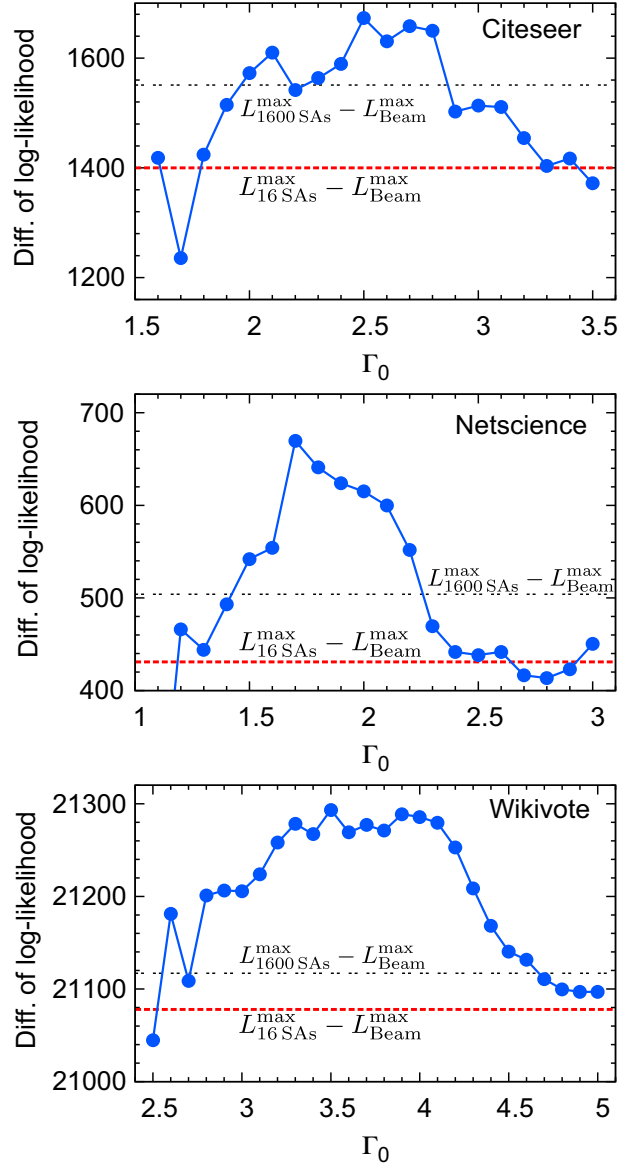}
\end{center}
\caption{
Experimental results for the maximum log-likelihoods (minimum energies) of QA and other algorithms.
Vertical axes show the difference of maximum log-likelihoods.
A higher value is better (closer to optimum states).
We denote $L_{\text{QA}}^{\text{max}}$, $L_{\text{16 SAs}}^{\text{max}}$ as the maximum log-likelihood of 16 CRPs in QA and SA, $L_{\text{1600 SAs}}^{\text{max}}$ as the maximum log-likelihood of 1600 CRPs in  SA, and $L_{\text{Beam}}^{\text{max}}$  as the maximum log-likelihood of the beam search with beam-width = $100$.
The solid line indicates $L_{\text{QA}}^{\text{max}}-L_{\text{Beam}}^{\text{max}}$.
Lower (red) and upper dotted lines indicate $L_{\text{16 SAs}}^{\text{max}}-L_{\text{Beam}}^{\text{max}}$ and $L_{\text{1600 SAs}}^{\text{max}}-L_{\text{Beam}}^{\text{max}}$.
When these lines take positive values, QA and SA outperform the beam search.
Whenever the solid line is higher than the dotted lines QA outperforms SAs.
The horizontal axes are the $\Gamma_0$ for QA .
}
\label{fig:results}
\end{figure}

\subsection{Experimental settings}
Our purpose is to search for a better MAP solution to a CRP in a small number of iterations (or short running time).
We evaluated  optimization algorithms in terms of maximum log-likelihood  because we want a state with the highest log-likelihood.
We compared QA with SA and the beam search.
We used the beam search with an inadmissible score function that achieved the best performance in \citep{Hal:AISTATS2007}.
We set the beam-width to $100$.
We did not  compare the variational inference with QA  because the variational inference cannot deal with the Chinese restaurant formulation of the Dirichlet process mixture.
That is,  it is hard to compare them because their objective functions are different.

Since we used a multi core processer with $16$ cores, we set $m=16$ (i.e., ran one CRP at one core)
We set $\alpha=1$ in the CRP.
$\alpha$ is easy to estimate in SAs and QA, but beam search cannot estimate it; therefore, we fixed it in these experiments.
The number of iterations, $T$, for SAs and QA was $30$.
We generated $16$ random seating arrangements $\{\sigma_j^{(\text{random})}\}_{j=1}^{16}$ for the initial settings  of QA and 16 SAs, i.e., we use $\sigma_j^{(\text{random})}$ for the same initial setting of the $j$-th seating arrangement.
Moreover, we compared QA ($m=16$) with 1600 SAs where their MAP solutions are the best one of 1600 simulations with different random initializations.
In 1600 SAs, we tried $100$ seeds and generated $m=16$ random seating arrangements for each seed, i.e., we ran the CRPs with $100 \times m(=16)=1600$ initial seating arrangements.

\subsection{Results and Discussions}
Figure \ref{fig:results} shows the experimental results.
QA and SAs outperformed the beam search because each line takes a positive value.
QA  finds a better local optimum than that of $16$ and $1600$ SAs at some $\Gamma_0$.
This means that it is useful to run QA with changing $\Gamma_0$ rather than to run multiple SAs.

The effective $\Gamma_0$ has a positive correlation with the number of nodes. For example, the effective $\Gamma_0$ is around $2$ in Netscience and is around $3.5$ in Wikivote, which has more nodes than that Netscience does. This is because the quantum effect term depends on $C \cdot f(\beta,\Gamma)$, where $C$ is the number of customers who share tables and thus depends on the number of customers (nodes). This means that the effective parameter range can be inferred from the number of nodes and we have only to check a few $\Gamma_0$ values.

QA needs more time and memory than SA with the linear order of $m$ because QA uses m CRPs
However, when a parallel processing environment can be used and we run multiple SAs in parallel, the scalability of QA is the same as that of SAs.
QA ($T=30$, $m=16$) and SA ($T=30$, $m=1$) took
about $15$ and $13$ seconds for Netscience,
about $25$ and $22$ seconds for Netscience, and
about $79$ and $76$ seconds for Wikivote, where each value was the averaged running time of a single simulation for SA.
Because of the multi core processing and the caching of customers sharing tables,
the running time of QA was almost the same as that of a single SA.
Therefore, QA makes the CRPs converge faster and finds a better seating arrangement than multiple-run SAs.
The estimated number of classes achieving the best performance in QA ($m=16$), $16$ SAs, $1600$ SAs and the beam search are $26$, $22$, $65$, and $61$ for Netscience,  $37$, $35$, $30$, and $57$ for Citeseer, and $8$, $8$, $8$, and $27$ for Wikivote.

We found that a small $\Gamma_0$ induces a fast schedule of $f$, which means $f \gg 0$ at small $\beta$.
The fast schedules make the convergence of QA too fast; therefore, QA converges at a worse optimum.
QA is similar to SAs at large $\Gamma_0$ because interaction $f$ remains at almost $0$ for a limited number of iterations.
Larger $\Gamma_0$ makes CRPs in QA more independent, which means the results of QA approach those of SA.
We found that interaction $f$ is almost zero when $\Gamma_0=5$ and $T=30$, which means that the performance of QA is similar to that of SA.
Therefore, in practice,  we only check values of $\Gamma_0$ in descending order from a large value of $\Gamma_0$, such as $\Gamma_0=5$.
That is, the effective value range is easy to infer from some $\Gamma_0$ values (in our experimental results, we show some non-effective values in order to provide the negative examples of QA).

\section{Conclusion}
We proposed a  QA algorithm for the DPM models based on the CRP.
Our algorithm is different from those of \cite{Kurihara:UAI2009} and \cite{Sato:UAI2009} in three ways: (i) it can handle an unfixed number of classes in mixture models, (ii) it does not require heuristics such as a purity, and (iii) it uses parallel processing in QA.
The proposed algorithm (Eq. \eqref{eq:py:qacrp}) is easy to implement because it is similar to a classical CRP (Eq. \eqref{eq:py:crp}).
That is, it is easy to apply the proposed algorithm to other nonparametric models with which it is not easy to apply beam search, such as an infinite relational model \citep{Kemp06}.
The proposed algorithm will therefore be a promising new optimization technique when it is used with  rapidly advancing multi core processers.
As shown in Eq.\eqref{qa_opt2}, our algorithm is regarded as an optimization with a regularized term and its performance depends on parameter $f$ like that other optimization algorithms with a regularized term does (e.g., L1 and L2 regularized optimization algorithms often used in machine learning).
For future work, it will be worth analyzing what kind of schedule of $f$ enables QA to work well.

\section*{Acknowledgements}
The present work was partially supported by the Mitsubishi Foundation,
and also by the Next Generation Super Computer Project, Nanoscience Program from MEXT of Japan.
The authors are partially supported by Grand-in-Aid for JSPS Fellows
(23-7601). Numerical calculations were partly performed on supercomputers at the Institute for Solid State Physics, University of
Tokyo.
This research was funded in part by the Joint Usage/Research Center for Interdisciplinary Large-scale Information Infrastructures in Japan.
This work was supported by a JSPS Grant-in-Aid for Young Scientists (B) 24700135.

\appendix
\section{Proof of Theorem \ref{theorem:qast}}\label{proof}

We use the following Trotter product formula \citep{Trotter59} to
approximate
\begin{align}
p_\text{QA}(\sigma;\beta,\Gamma)= \frac{1}{Z}\sigma^\top e^{- \beta (\cH_\text{c} + \cH_\text{q})} \sigma. \label{eq:qa_model}
\end{align}

If $A_1,...,A_L$ are symmetric matrices, we have
\begin{align}
  \exp\left( \sum_{l=1}^L A_l \right)
  \! = \!
  \left[ \prod_{l=1}^L \exp(A_l/m) \right]^{\! m}
  \!\!\!\!\!
  + O\left(  \frac{1}{m}  \right)
  .
\end{align}
Applying the Trotter product formula to Eq.\eqref{eq:qa_model} with
finite $m$, we have
\begin{align}
  p_\text{QA}(\sigma; \beta, \Gamma)
  &
  =
  \frac{1}{Z}
  \sigma^\top e^{- \beta (\cH_\text{c} + \cH_\text{q})} \sigma
  \nn\\
  &
  \approx
  \frac{1}{Z}
  \sigma^\top \left(
  e^{- \frac{\beta}{m} \cH_\text{c}}
  e^{- \frac{\beta}{m} \cH_\text{q}}
  \right)^m
  \sigma
  \label{eq:trotter_expansion}
  .
\end{align}
We evaluate the residual of this approximation.
 Since $e^{A_1+A_2} = e^{A_1} e^{A_2}$ does not hold in
 general\footnote{If $A_1 A_2 = A_2 A_1$, then $e^{A_1 + A_2} = e^{A_1}
   e^{A_2}$.}, we need to use the Trotter product formula for computation.
We rewrite $\sigma$ as $\sigma_1$ and note that $\sigma_1^\top
\!\left(e^A\right)^{\!2}\! \sigma_1 = \sum_{\sigma_2} \sigma_1^\top e^A
\sigma_2 \sigma_2^\top e^A \sigma_1$.  Hence, we express Eq.\eqref{eq:trotter_expansion} by marginalizing out auxiliary variables
$\{\sigma_1',\sigma_2,\sigma_2',...,\sigma_m,\sigma_m'\}$,
\begin{align}
  &
  \sigma_1^\top
  \left(
  e^{- \frac{\beta}{m} \cH_\text{c}}
  e^{- \frac{\beta}{m} \cH_\text{q}}
  \right)^m
  \sigma_1
  \nn\\
  &=
  \sum_{\sigma_1'}
  \sum_{\sigma_2}
  ...
  \sum_{\sigma_m}
  \sum_{\sigma_m'}
  \sigma_1^\top
  e^{- \frac{\beta}{m} \cH_\text{c}}
  \sigma_1'
  \sigma_1'^\top
  e^{- \frac{\beta}{m} \cH_\text{q}}
  \sigma_{2}
  \nn\\
  &\quad\quad\cdots
  \sigma_m^\top
  e^{- \frac{\beta}{m} \cH_\text{c}}
  \sigma_m'
  \sigma_m'^\top
  e^{- \frac{\beta}{m} \cH_\text{q}}
  \sigma_{m+1}
  \nn\\
  &=
  \sum_{\sigma_1'}
  \sum_{\sigma_2}
  ...
  \sum_{\sigma_m}
  \sum_{\sigma_m'}
  \prod_{j=1}^m
  \sigma_j^\top
  e^{- \frac{\beta}{m} \cH_\text{c}}
  \sigma_j'
  \sigma_j'^\top
  e^{- \frac{\beta}{m} \cH_\text{q}}
  \sigma_{j+1}
  ,
  \label{eq:qa_naive_st_expansion}
\end{align}
where $\sigma_{m+1} = \sigma_1$.
To express Eq.\eqref{eq:qa_naive_st_expansion} more particularly, we use the
following Lemma \ref{lemma:trace_classical} and Lemma \ref{lemma:trace_quantum}.

\begin{lemma}
  \label{lemma:trace_classical}
\begin{align}
  \sigma_j^\top
  e^{- \frac{\beta}{m} \cH_\text{c}}
  \sigma_j'
  =&
  \exp\left[- \frac{\beta}{m} E(\sigma_j) \right]
  \delta(\sigma_j, \sigma_j'),
  \label{eq:trace_classical}
\end{align}
where $\delta(\sigma_j, \sigma_j') = 1$ if $\sigma_j=\sigma_j'$ and
$\delta(\sigma_j, \sigma_j') = 0$ otherwise.
\end{lemma}
\begin{proof}
  By the definition, $e^{- \frac{\beta}{m} \cH_\text{c}}$ is diagonal
  with $[e^{- \frac{\beta}{m} \cH_\text{c}}]_{kk} = E(\sigma^{(k)})$,
  and $\sigma_j$ and $\sigma_j'$ are binary indicator vectors,
  i.e. only one element in $\sigma_j$ is one and the others are zero.
  Thus, the above lemma holds.
\end{proof}

\begin{lemma}
  \label{lemma:trace_quantum}
  \begin{align}
    &\sigma_j'^\top
    e^{- \frac{\beta}{m} \cH_\text{q}}
    \sigma_{j+1}
  	= \left [ \sinh \left(\frac{\beta}{m}\Gamma \right)\right]^{2N}\nn\\
  	&~~~\exp \left[\sum_{i=1}^{N}\sum_{n=1}^{N}\delta( \tilde{\sigma'}_{j,i,n},\tilde{\sigma}_{j+1,i,n}) \log \coth \left(\frac{\beta}{m}\Gamma \right)	\right]. \label{lemma:trace_quantum:eq}
\end{align}
\end{lemma}

\begin{proof}

Using $(A\otimes B)(C \otimes D) = (AC) \otimes (BD)$
and $e^{A_1 + A_2} = e^{A_1}e^{A_2}$ when $A_1A_2 = A_2A_1$, we find,
\begin{align}
  \sigma_j'^\top
  e^{- \frac{\beta}{m} \cH_\text{q}}
  \sigma_{j+1}
  =&
  \sigma_j'^\top
  e^{- \frac{\beta}{m} \{-\Gamma \sum_{i=1}^N \sum_{n=1}^N  \sigma_{i,n}^x\}}
  \sigma_{j+1}
  \nn\\
  =&
  \sigma_j'^\top
  \left(
  \bigotimes_{i=1}^N
  \bigotimes_{n=1}^N
  e^{\frac{\beta}{m}\Gamma \sigma_{i,n}^x}
  \right)
  \sigma_{j+1}
  \nn\\
  =&
  \prod_{i=1}^N
  \prod_{n=1}^N
  \tilde{\sigma'}_{j,i,n}^\top
  e^{\frac{\beta}{m}\Gamma \sigma^x}
  \tilde{\sigma}_{j+1,i,n}.
  \label{eq:trace_quantum_after_Minka}
\end{align}
Note that  $\sigma_j = \bigotimes_{i=1}^N \bigotimes_{n=1}^N \tilde{\sigma}_{j,i,n}$.

\begin{align}
  e^{\frac{\beta}{m}\Gamma \sigma^x}
  =&
  \sum_{l=0}^\infty \frac{1}{l!}
  \left(\frac{\beta}{m}\Gamma \sigma^x \right)^l
  \nn\\
  =&
  \left[1+\frac{1}{2!} \left (\frac{\beta}{m}\Gamma \right )^2+\cdots \right ]\bbE
  \nn\\
  &+
  \left[\frac{\beta}{m}\Gamma +\frac{1}{3!} \left (\frac{\beta}{m}\Gamma \right )^3+\cdots \right ]\sigma^{x}
  \nn\\
  =&
  \cosh \left(\frac{\beta}{m}\Gamma \right) \bbE
  +
  \sinh \left(\frac{\beta}{m}\Gamma \right) \sigma^{x}.
  \label{eq:trace_quantum_support1}
\end{align}

Substituting Eq.\eqref{eq:trace_quantum_support1} into Eq.\eqref{eq:trace_quantum_after_Minka}, we have Eq.\eqref{lemma:trace_quantum:eq} because

\begin{align}
  &\tilde{\sigma'}_{j,i,n}^\top
  \left [
  \cosh \left(\frac{\beta}{m}\Gamma \right) \bbE
  +
  \sinh \left(\frac{\beta}{m}\Gamma \right) \sigma^{x}
  \right ]
  \tilde{\sigma}_{j+1,i,n}
  \nn\\
  &=
  \cosh \left(\frac{\beta}{m}\Gamma \right) \delta( \tilde{\sigma'}_{j,i,n},\tilde{\sigma}_{j+1,i,n})\nn\\
  &~~~+
  \sinh \left(\frac{\beta}{m}\Gamma \right) (1-\delta( \tilde{\sigma'}_{j,i,n},\tilde{\sigma}_{j+1,i,n}))
\end{align}
\end{proof}

Using Lemma \ref{lemma:trace_classical} and Lemma \ref{lemma:trace_quantum} into Eq.\eqref{eq:qa_naive_st_expansion},

\begin{align}
  &
 \sigma_1^\top
  \left(
  e^{- \frac{\beta}{m} \cH_\text{c}}
  e^{- \frac{\beta}{m} \cH_\text{q}}
  \right)^m
  \sigma_1
  \nn\\
  =&
  \sum_{\sigma_2}
  \cdots
  \sum_{\sigma_m}
  \prod_{j=1}^m
  \exp\left[- \frac{\beta}{m} E(\sigma_j) \right]
  \left[\sinh \left(\frac{\beta}{m}\Gamma \right)\right]^{2N}
  \nn\\
  &\exp \left[\sum_{i=1}^{N}\sum_{n=1}^{N}\delta( \tilde{\sigma'}_{j,i,n},\tilde{\sigma}_{j+1,i,n}) \log \coth \left(\frac{\beta}{m}\Gamma \right)	\right].
  \label{eq:qa_naive_st_expansion_no_primes}
\end{align}
and from Eq.\eqref{eq:trotter_expansion},  we have shown,
\begin{align}
  &p_\text{QA}(\sigma_1; \beta, \Gamma)
  \approx
  \sum_{\sigma_2}
  \cdots
  \sum_{\sigma_m}
  p_\text{QA-ST}(\sigma_1,\sigma_2,...,\sigma_m; \beta, \Gamma)
  \label{eq:incomplete_approximated_qa}
  .
\end{align}
The same relation holds for $\sigma_2, \sigma_3, \cdots, \sigma_m$.

\section{Derivation of Eq.\eqref{eq:py:qacrp}}\label{dev:qacrp}

Since
\begin{align}
&p_\text{QA}(\tilde \sigma_{j,i} | \{ \sigma_{d}\}_{d=1}^{m} \backslash \{ \tilde \sigma_{j,i} \};\beta,\Gamma) \nn\\
&\propto e^{- \frac{\beta}{m} E(\sigma_j)}e^{f(\beta,\Gamma)(s(\sigma_{j-1},\sigma_{j}) + s(\sigma_j,\sigma_{j+1}))},
\end{align}
 we have

\begin{align}
&p_\text{QA}(\tilde \sigma_{j,i} | \{ \sigma_{d}\}_{d=1}^{m} \backslash \{ \tilde \sigma_{j,i} \};\beta,\Gamma) \propto\nn\\
&\begin{cases}
\displaystyle \left ( \frac{ \sum_{n=1,n \not = i}^{N} \delta (\tilde \sigma_{j,i,n},(1,0)^{\top}) }{N-1 + \alpha}\right )^{\frac{\beta}{m}} e^{\tilde s(\tilde \sigma_{j-1,i}, \tilde \sigma_{j,i},\tilde \sigma_{j+1,i})f(\beta,\Gamma)},
\\
~~~~~~~~~~~~~~~~~~~~~~~~~~~~~~~~~~~~~\text{if $\tilde \sigma_{j,i} \in \tilde \Sigma_{j,i}\backslash \{\rho_{2N-1}^{(i)}\}$},\\
\displaystyle \left ( \frac{\alpha}{N-1+\alpha} \right )^{\frac{\beta}{m}},
~~~~~~~~~~~~~~\text{if $\tilde \sigma_{j,i}=\rho_{2N-1}^{(i)}$},\\
0~~~~~~~~~~~~~~~~~~~~~~~~~~~~~~~~~~~\text{if $\tilde \sigma_{j,i} \not \in \tilde \Sigma_{j,i}$},
\label{eq:qacrp_gibbs}
\end{cases}
\end{align}
\begin{align}
&\tilde s(\tilde \sigma_{j-1,i},\tilde \sigma_{j,i},\tilde \sigma_{j+1,i})\nn\\
&=\sum_{n=1,n \not= i}^{N}\delta( \tilde{\sigma}_{j-1,i,n},\tilde{\sigma}_{j,i,n}) + \sum_{n=1,n \not= i}^{N}\delta( \tilde{\sigma}_{j,i,n},\tilde{\sigma}_{j+1,i,n}).
\end{align}
This is easy to understand when you consider the meaning of $\tilde s(\tilde \sigma_{j-1,i},\tilde \sigma_{j,i},\tilde \sigma_{j+1,i})$ in multiple CRPs.
$\tilde s(\tilde \sigma_{j-1,i},\tilde \sigma_{j,i},\tilde \sigma_{j+1,i})$ indicates the number of customers who share tables with the $i$-th customer in the $j-1$-th and $j+1$-th CRPs.
Therefore,  Eq.\eqref{eq:py:qacrp} is derived as another formulation of Eq.\eqref{eq:qacrp_gibbs}.
Note that Eq.\eqref{eq:py:qacrp} is the approximation of  Eq.\eqref{eq:gibbs:dm-crp-qa}.

\section{Details of Network Model}\label{appdx:network_model}
In this section, we explain the Newman network model.
$\mathscr{V}$ is the vertex set.
$v$ is a vertex; i.e., $v \in \mathscr{V}$.
$V$ is the number of vertices.
$K$ is the number of classes.
Suppose that the vertices fall into $K$ classes with probability $\bpi$, where $\pi_k$ is the probability that a vertex is assigned to class $k$. Vertex $i$ belongs to class $k$, indicated by $z_i=k$. Each class has a probability $\phi_{kv}$ that a link from a particular vertex in class $k$ connects to vertex $v$. A link from vertex $i$ to vertex $v$ is indicated by $\ell_i=v$. Each vertex links to other vertices in accordance with $\boldsymbol{\phi}$. That is, vertex $i$ links to vertex $v$ in accordance with $\phi_{z_i v}$. The generation process for link $\ell_i$ is represented by
$
\ell_i \sim \text{Multi}(\phi_{z_i}),~z_i\sim\text{Multi}(\bpi),
$
where $\phi_{z_i}=(\phi_{z_i1},\phi_{z_i2},\cdots,\phi_{z_iV})$, and $\text{Multi}(\cdot)$ is a multinomial distribution.

Suppose that $\phi_t$ is distributed in accordance with the Dirichlet distribution $H(\tau)$; i.e., $\phi_k \sim H(\tau)$, where $\tau$ is a parameter of the Dirichlet distribution. $G$ is a random probability measure over $\phi$:
$
G \sim \text{DP}(\alpha,H(\tau)),
$
where $\text{DP}(\cdot)$ indicates the Dirichlet process (DP), $\alpha$ is the DP concentration parameter that is equal to the hyper parameter of the CRP, and $H$ is the base measure, which is the Dirichlet distribution here. The generation process for link $\ell_i$ is represented by
$
\ell_i \sim \text{Multi}(\phi_{z_i}),~\phi_{z_i} \sim G.
$

Here, we define $A$ as an adjacency matrix with elements $A_{iv} = 1$ if there is an edge from $i$ to $v$; otherwise $ A_{iv} = 0$.
The probability of $z_i$ given $z_{-i} = \{z\}\backslash z_i$ and adjacency matrix $A$ is
\begin{align}
&p(z_i=k| A, z_{-j};\alpha)\propto\nn\\
&~p(A_i|z_i=k,A_{-i},z_{-i}) p(z_i=k|z_{-i};\alpha),
\label{dp_newman_post}
\end{align}
where $A_i=(A_{i1},A_{i2},\cdots,A_{iV})$, and $A_{-i}={A}\backslash A_i$.

We can calculate the probability of Eq.\eqref{dp_newman_post} as follows.
\begin{align}
&p(A_i|z_i=k,A_{-i},z_{-i})=\nn \\
&\frac{ g(V\tau + \sum_{u\not=i}\sum_{v}A_{uv}z_{u}^{k} )}{ g(V\tau + \sum_{u}\sum_{v}A_{uv}z_{u}^{k}) } \prod_{v} \frac{ g(\tau + \sum_{u}A_{uv}z_{u}^{k})}{ g(\tau + \sum_{u\not=j}A_{uv}z_{u}^{k})},
\end{align}
where $g(\cdot)$ is the gamma function, and $z_{i}^{k}$ indicates $1$ if $z_{i}={k}$; otherwise, it indicates $0$.
\begin{align}
p(z_i=k|z_{-i};\alpha) =
\begin{cases} \displaystyle \frac{ N_{k}^{-i} }{V -1 + \alpha } & \mbox{(if $k$ previously used.)} \\
\displaystyle \frac{\alpha}{V-1+\alpha} & \mbox{(if $k$ is new.)}
\end{cases},
\end{align}
where $N_{k}^{-i}$ is the number of vertices except vertex $i$ assigned to class $k$; i.e., $N_{k}^{-i} = \sum_{v\not=i}{z_{v}^{k}}$.
We can adapt a Gibbs sampler for estimating $z_i$ by using Eq. \eqref{dp_newman_post}.

\end{document}

%% file: macros.tex
\usepackage{amsmath}
\usepackage{amsfonts}
\usepackage{bbm} 
\DeclareMathOperator*{\argmax}{argmax}

\DeclareMathOperator*{\Tr}{Tr}

\newcommand{\bbE}{\mathbb{E}}

\newcommand{\cH}{\mathcal{H}}

\newcommand{\bB}{{\boldsymbol{B}}}

\newcommand{\bZ}{{\boldsymbol{Z}}}

\newcommand{\bx}{{\boldsymbol{x}}}

\newcommand{\bpi}{{\boldsymbol{\pi}}}

\newcommand{\nn}{\nonumber}

\def\str1#1#2{\left[ \begin{array}{c} #1 \\ #2 \end{array}\right]}

%% file: macros_theorems.tex
\usepackage{amsthm} 

\newtheorem{theorem}{Theorem}[section]
\newtheorem{lemma}[theorem]{Lemma}


%% file: QACRP_20130519_IsseiSato.bbl
\begin{thebibliography}{23}
\expandafter\ifx\csname natexlab\endcsname\relax\def\natexlab#1{#1}\fi
\expandafter\ifx\csname url\endcsname\relax
  \def\url#1{\texttt{#1}}\fi
\expandafter\ifx\csname urlprefix\endcsname\relax\def\urlprefix{URL }\fi
\providecommand{\eprint}[2][]{\url{#2}}
\providecommand{\bibinfo}[2]{#2}
\ifx\xfnm\relax \def\xfnm[#1]{\unskip,\space#1}\fi
\bibitem[{Aldous(1985)}]{Aldous1985}
\bibinfo{author}{Aldous, D.}, \bibinfo{year}{1985}.
\newblock \bibinfo{title}{Exchangeability and related topic}.
\newblock \bibinfo{journal}{Ecole d'Ete de Probabilities de Saint-Flour}
  \bibinfo{volume}{XIII-1983}, \bibinfo{pages}{1--198}.
\bibitem[{Antoniak(1974)}]{Antoniak74}
\bibinfo{author}{Antoniak, C.}, \bibinfo{year}{1974}.
\newblock \bibinfo{title}{Mixtures of {D}irichlet processes with applications
  to bayesian nonparametric problems}.
\newblock \bibinfo{journal}{The Annals of Statistics} \bibinfo{volume}{2},
  \bibinfo{pages}{1152--1174}.
\bibitem[{Blei and Jordan(2005)}]{Blei05}
\bibinfo{author}{Blei, D.M.}, \bibinfo{author}{Jordan, M.I.},
  \bibinfo{year}{2005}.
\newblock \bibinfo{title}{Variational inference for dirichlet process
  mixtures}.
\newblock \bibinfo{journal}{Bayesian Analysis} \bibinfo{volume}{1},
  \bibinfo{pages}{121--144}.
\bibitem[{DaumeIII(2007)}]{Hal:AISTATS2007}
\bibinfo{author}{DaumeIII, H.}, \bibinfo{year}{2007}.
\newblock \bibinfo{title}{Fast search for dirichlet process mixture models},
  in: \bibinfo{booktitle}{Proceedings of Artificial Intelligence and
  Statistics}.
\bibitem[{Dempster et~al.(1977)Dempster, Laird and
  Rubin}]{DempsterLairdRubin77}
\bibinfo{author}{Dempster, A.}, \bibinfo{author}{Laird, N.},
  \bibinfo{author}{Rubin, D.}, \bibinfo{year}{1977}.
\newblock \bibinfo{title}{Maximum likelihood from incomplete data via the {EM}
  algorithm}.
\newblock \bibinfo{journal}{Journal of the Royal Statistical Society series B}
  \bibinfo{volume}{39}, \bibinfo{pages}{1--38}.
\bibitem[{Farhi et~al.(2001)Farhi, Goldstone, Gutmann, J.~Lapan and
  Preda.}]{Farhi2001}
\bibinfo{author}{Farhi, E.}, \bibinfo{author}{Goldstone, J.},
  \bibinfo{author}{Gutmann, S.}, \bibinfo{author}{J.~Lapan, A.L.},
  \bibinfo{author}{Preda., D.}, \bibinfo{year}{2001}.
\newblock \bibinfo{title}{A quantum adiabatic evolution algorithm applied to
  random instances of an np -complete problem}.
\newblock \bibinfo{journal}{Science} \bibinfo{volume}{292},
  \bibinfo{pages}{472--476}.
\bibitem[{Frey and Dueck(2007)}]{FreyDueck:Science2007}
\bibinfo{author}{Frey, B.J.}, \bibinfo{author}{Dueck, D.},
  \bibinfo{year}{2007}.
\newblock \bibinfo{title}{Clustering by passing messages between data points}.
\newblock \bibinfo{journal}{Science} \bibinfo{volume}{315},
  \bibinfo{pages}{972--976}.
\bibitem[{Geman and Geman(1984)}]{Geman84}
\bibinfo{author}{Geman, S.}, \bibinfo{author}{Geman, D.}, \bibinfo{year}{1984}.
\newblock \bibinfo{title}{Stochastic relaxation, {G}ibbs distributions, and the
  {B}ayesian restoration of images}.
\newblock \bibinfo{journal}{IEEE Pattern Analysis and Machine Intelligence,}
  \bibinfo{volume}{6}, \bibinfo{pages}{721--741}.
\bibitem[{Kadowaki and Nishimori(1998)}]{Kadowaki98}
\bibinfo{author}{Kadowaki, T.}, \bibinfo{author}{Nishimori, H.},
  \bibinfo{year}{1998}.
\newblock \bibinfo{title}{Quantum annealing in the transverse ising model}.
\newblock \bibinfo{journal}{Physical Review E} \bibinfo{volume}{58},
  \bibinfo{pages}{5355--5363}.
\bibitem[{Kemp et~al.(2006)Kemp, Tenenbaum, Griffiths, Yamada and
  Ueda}]{Kemp06}
\bibinfo{author}{Kemp, C.}, \bibinfo{author}{Tenenbaum, J.B.},
  \bibinfo{author}{Griffiths, T.L.}, \bibinfo{author}{Yamada, T.},
  \bibinfo{author}{Ueda, N.}, \bibinfo{year}{2006}.
\newblock \bibinfo{title}{Learning systems of concepts with an infinite
  relational model}, in: \bibinfo{booktitle}{AAAI}.
\bibitem[{Kirkpatrick et~al.(1983)Kirkpatrick, Gelatt and
  Vecchi}]{Kirkpatrick83}
\bibinfo{author}{Kirkpatrick, S.}, \bibinfo{author}{Gelatt, C.D.},
  \bibinfo{author}{Vecchi, M.P.}, \bibinfo{year}{1983}.
\newblock \bibinfo{title}{Optimization by simulated annealing}.
\newblock \bibinfo{journal}{Science} \bibinfo{volume}{220},
  \bibinfo{pages}{671--680}.
\bibitem[{Kurihara et~al.(2009)Kurihara, Tanaka and
  Miyashita}]{Kurihara:UAI2009}
\bibinfo{author}{Kurihara, K.}, \bibinfo{author}{Tanaka, S.},
  \bibinfo{author}{Miyashita, S.}, \bibinfo{year}{2009}.
\newblock \bibinfo{title}{Quantum annealing for clustering}, in:
  \bibinfo{booktitle}{Proceedings of the 25th Conference on Uncertainty in
  Artificial Intelligence}.
\bibitem[{Kurihara et~al.(2007)Kurihara, Welling and Teh}]{Kurihara07IJCAI}
\bibinfo{author}{Kurihara, K.}, \bibinfo{author}{Welling, M.},
  \bibinfo{author}{Teh, Y.W.}, \bibinfo{year}{2007}.
\newblock \bibinfo{title}{Collapsed variational dirichlet process mixture
  models}, in: \bibinfo{booktitle}{IJCAI}, pp. \bibinfo{pages}{2796--2801}.
\bibitem[{L.~Xu and Oja(1993)}]{KrzyzakOja:NN1993}
\bibinfo{author}{L.~Xu, A.K.}, \bibinfo{author}{Oja, E.}, \bibinfo{year}{1993}.
\newblock \bibinfo{title}{Rival penalized competitive learning for clustering
  analysis, rbf net, and curve detection}.
\newblock \bibinfo{journal}{IEEE Transaction of Neural Network}
  \bibinfo{volume}{4}, \bibinfo{pages}{636--649}.
\bibitem[{Lloyd(1996)}]{Lloyd1996}
\bibinfo{author}{Lloyd, S.}, \bibinfo{year}{1996}.
\newblock \bibinfo{title}{Universal quantum simulators}.
\newblock \bibinfo{journal}{Science} \bibinfo{volume}{273},
  \bibinfo{pages}{1073--1078}.
\bibitem[{Neal(2000)}]{Neal00}
\bibinfo{author}{Neal, R.M.}, \bibinfo{year}{2000}.
\newblock \bibinfo{title}{Markov chain sampling methods for dirichlet process
  mixture models}.
\newblock \bibinfo{journal}{Journal of Computational and Graphical Statistics}
  \bibinfo{volume}{9}, \bibinfo{pages}{249--265}.
\bibitem[{Newman and Leicht(2007)}]{newman:2007}
\bibinfo{author}{Newman, M.E.J.}, \bibinfo{author}{Leicht, E.A.},
  \bibinfo{year}{2007}.
\newblock \bibinfo{title}{Mixture models and exploratory analysis in networks},
  in: \bibinfo{booktitle}{Proceedings of National Academy of Sciences of the
  United States of America}.
\bibitem[{Nielsen and Chuang(2000)}]{NielsenChuang2000}
\bibinfo{author}{Nielsen, M.}, \bibinfo{author}{Chuang, I.},
  \bibinfo{year}{2000}.
\newblock \bibinfo{title}{Quantum Computation and Quantum Information}.
\newblock \bibinfo{publisher}{Cambridge University Press}.
\bibitem[{Santoro et~al.(2002)Santoro, Marto\v{n}\'{a}k, Tosatti and
  Car}]{Santoro02}
\bibinfo{author}{Santoro, G.E.}, \bibinfo{author}{Marto\v{n}\'{a}k, R.},
  \bibinfo{author}{Tosatti, E.}, \bibinfo{author}{Car, R.},
  \bibinfo{year}{2002}.
\newblock \bibinfo{title}{Theory of quantum annealing of an {I}sing spin
  glass}.
\newblock \bibinfo{journal}{Science} \bibinfo{volume}{295},
  \bibinfo{pages}{2427--2430}.
\bibitem[{Sato et~al.(2009)Sato, Kurihara, Tanaka, Nakagawa and
  Miyashita}]{Sato:UAI2009}
\bibinfo{author}{Sato, I.}, \bibinfo{author}{Kurihara, K.},
  \bibinfo{author}{Tanaka, S.}, \bibinfo{author}{Nakagawa, H.},
  \bibinfo{author}{Miyashita, S.}, \bibinfo{year}{2009}.
\newblock \bibinfo{title}{Quantum annealing for variational bayes inference},
  in: \bibinfo{booktitle}{Proceedings of the 25th Conference on Uncertainty in
  Artificial Intelligence}.
\bibitem[{Suzuki(1976)}]{Suzuki76}
\bibinfo{author}{Suzuki, M.}, \bibinfo{year}{1976}.
\newblock \bibinfo{title}{Relationship between $d$-dimensional quantal spin
  systems and $(d+1)$-dimensional {I}sing systems -- equivalence , critical
  exponents and systematic approximants of the partition function and spin
  correlations --}.
\newblock \bibinfo{journal}{Progress of Theoretical Physics}
  \bibinfo{volume}{56}, \bibinfo{pages}{1454--1469}.
\bibitem[{Trotter(1959)}]{Trotter59}
\bibinfo{author}{Trotter, H.F.}, \bibinfo{year}{1959}.
\newblock \bibinfo{title}{On the product of semi-groups of operators}.
\newblock \bibinfo{journal}{Proceedings of the American Mathematical Society}
  \bibinfo{volume}{10}, \bibinfo{pages}{545--551}.
\bibitem[{Wang and Lai(2011)}]{WangLai:NC2011}
\bibinfo{author}{Wang, C.}, \bibinfo{author}{Lai, J.}, \bibinfo{year}{2011}.
\newblock \bibinfo{title}{Energy based competitive learning}.
\newblock \bibinfo{journal}{Neurocomputing} \bibinfo{volume}{74},
  \bibinfo{pages}{2265--2275}.
\end{thebibliography}
